\documentclass[final,3p,times]{elsarticle}

\usepackage{lineno,hyperref}
\usepackage{amsthm,amsmath,amssymb,eucal}
\journal{Physics Letters A}

\newtheorem{theorem}{Theorem}
\newtheorem{lemma}[theorem]{Lemma}
\newdefinition{remark}{Remark}

\begin{document}

\begin{frontmatter}

\title{Smilansky-Solomyak model with a $\delta'$-interaction}
%\tnotetext[mytitlenote]{Fully documented templates are available in the elsarticle package on \href{http://www.ctan.org/tex-archive/macros/latex/contrib/elsarticle}{CTAN}.}

%% Group authors per affiliation:
%\author{Elsevier\fnref{myfootnote}}
%\address{Radarweg 29, Amsterdam}
%\fntext[myfootnote]{Since 1880.}

%% or include affiliations in footnotes:
%\author[mymainaddress,mysecondaryaddress]{Elsevier Inc}
%\ead[url]{www.elsevier.com}

%\author[mysecondaryaddress]{Global Customer Service\corref{mycorrespondingauthor}}
%\cortext[mycorrespondingauthor]{Corresponding author}
%\ead{support@elsevier.com}

%\address[mymainaddress]{1600 John F Kennedy Boulevard, Philadelphia}
%\address[mysecondaryaddress]{360 Park Avenue South, New York}

\author[label1,label2]
{Pavel Exner}
\ead{exner@ujf.cas.cz}
\author[label2,label3]
{Ji\v{r}\'{\i} Lipovsk\'{y}\corref{cor1}}
\ead{jiri.lipovsky@uhk.cz}
\address[label1]{Doppler Institute for Mathematical Physics and Applied Mathematics, Czech Technical University,
B\v rehov{\'a} 7, 11519 Prague, Czechia}
\address[label2]{Department of Theoretical Physics, Nuclear Physics Institute, Czech Academy of Sciences, 25068 \v{R}e\v{z} near Prague, Czechia}
\address[label3]{Department of Physics, Faculty of Science, University of Hradec Kr\'alov\'e, Rokitansk\'eho 62, 500 03 Hradec Kr\'alov\'e, Czechia}

\begin{abstract}
\noindent We investigate a strongly singular version of the model of irreversible dynamics proposed by Smilansky and Solomyak in which the interaction responsible for an abrupt change of the spectrum is of $\delta'$ type. We determine the spectrum in both the subcritical and supercritical regimes and discuss its character as well as its asymptotic properties of the discrete spectrum in terms of the coupling constant.
\end{abstract}


\end{frontmatter}

%\linenumbers

%%%%%%%%%%%%%%%%%%%%%%
\section{Introduction}
%%%%%%%%%%%%%%%%%%%%%%

Solutions of the Schr\"odinger equation describing an isolated system are time-reversal. We know, however, that irreversible processes are ubiquitous, the reason being that real physical systems usually interact with the environment. This is typically described through a coupling between the system Hamiltonian and that of a thermal bath. In most situations that one encounters the latter is a `large' system with an infinite number of degrees of freedom, the bath Hamiltonian has a continuous spectrum, and the presence (or absence) of irreversible modes is determined by the
energies involved rather than the coupling strength. To show that in general neither of those assumptions need to be true, Uzy Smilansky \cite{Smi1} proposed a model of a quantum graph coupled to the `bath' which may consist of  one-dimensional harmonic oscillators, in the extreme case even a single one. He showed that if the coupling exceeds a critical value, such a coupled system exhibits an irreversible behavior.

The model, in the simplest version describing a $\delta$-coupling between the Schr\"odinger operator on a line  and a harmonic oscillator with the coupling strength dependent on the oscillator variable, was later rigorously studied and generalized by Michael Solomyak and coauthors \cite{Sol4, Sol1, EvSo1, EvSo2,Sol2, Sol3, NS, RoSo} and in other papers to be mentioned a little below. What is important, in Solomyak's interpretation the model describes a two-dimensional
Schr\"odinger particle interacting with a potential composed of a regular (harmonic) and a singular part. The reason to mention that is that while in the original model the two ways of looking at the system are equivalent, it may not be so in some generalizations.

The analysis done by Solomyak and coauthors focused on spectral properties of the model, in particular, on the abrupt change they exhibit when the coupling constant exceeds a critical value. Guarneri \cite{Gu} examined the (slightly modified) model from the dynamical point of view and showed that the irreversibility means that in the supercritical regime the wave packet may escape to infinity along the singular potential `channel'. He also asked whether the model could have a version in which the $\delta$ potential is replaced by a regular one. The affirmative answer was provided in \cite{BaE1} where such a potential family was constructed, in \cite{BaE2} it was demonstrated that the effect persists even if the system is exposed to a homogeneous magnetic field (in which case the original Smilansky interpretation is ultimately lost). Moreover, it was shown that the original model has a rich resonance structure \cite{ELT2, ELT1}.

The aim of the present paper is to illustrate that the switch between different spectral regimes, associated with the passage from the `reversible' to `irreversible' behavior, is a robust effect which survives not only the above mentioned `regularization' but also a modification in the opposite direction consisting of replacing the $\delta$ potential by a more singular coupling. We shall consider the Hamiltonian formally written as
 % -------------- %
$$
  \mathbf{H}_\beta = -\frac{\partial^2}{\partial x^2} + \frac{1}{2}\left(-\frac{\partial^2}{\partial y^2}+y^2\right) + \frac{\beta}{y}\, \delta'(x)\,,
$$
 % -------------- %
where the last term is a shorthand for the interaction defined rigorously through the boundary conditions \eqref{eq-cc1}, \eqref{eq-cc2} below, cf.~Chap.~I.4 in \cite{AGHH}. We assume $\beta>0$ since $-\beta$ gives the same spectrum as $\beta$, as explained in Section~\ref{sec:description} below. We are going to show that the critical value for the spectral switch is now $\beta=2\sqrt{2}$. For $\beta>2\sqrt{2}$ the absolutely continuous spectrum coincides with the interval  $(\textstyle{\frac12},\infty)$ and there is a non-empty discrete spectrum in the interval $(0,\frac12)$. On the other hand, for $0<\beta<2\sqrt{2}$ the spectrum is absolutely continuous and covers the whole real axis; in the critical case, $\beta=2\sqrt{2}$, the spectrum is again purely absolutely continuous and covers the interval $[0,\infty)$.

Furthermore we shall investigate the asymptotic behavior of the discrete spectrum. We will show that for $\beta$ large there is a single eigenvalue for which we provide an approximate expression. On the other hand, the number of eigenvalues increases as $\beta\to 2\sqrt{2}+$ and we will describe the way in which they accumulate. The main results will be stated in the next section, the rest of the paper is devoted to their proofs. Our task will be simplified by the fact that at several places we will be able to employ the same arguments as used in the $\delta$ case borrowing them from \cite{Sol1, NS}.

Before proceeding further, a comment is due on the notion of the interaction strength. In contrast to the $\delta$ potential which is `natural' in the sense it can be regarded as a sharply localized potential well or barrier \cite[Sec.~I.3.2]{AGHH}, the $\delta'$ interaction is a much more involved object \cite{ENZ1} and one can speak of its strength only with some license. Being defined in the common way, a single attractive $\delta'$ interaction on the line has the eigenvalue $-\frac{4}{\beta^2}$, cf. \cite[Thm.~I.4.3]{AGHH}, which motivates us to say that it is \emph{weak} for $\beta$ \emph{large} and vice versa.

%%%%%%%%%%%%%%%%%%%%%%%%%%%%%%%%%%%%%%%%%%%%%%%%%%%
\section{Description of the model and main results}\label{sec:description}
%%%%%%%%%%%%%%%%%%%%%%%%%%%%%%%%%%%%%%%%%%%%%%%%%%%

Our first task is to define properly the Hamiltonian. It will be the operator in $L^2(\mathbb{R}^2)$ corresponding to the differential expression
 % -------------- %
$$
  \mathbf{H}_\beta \Psi(x,y)= - \frac{\partial^2 \Psi}{\partial x^2}(x,y) +\frac{1}{2}\left(- \frac{\partial^2 \Psi}{\partial y^2}(x,y) +y^2\Psi(x,y)\right)
$$
 % -------------- %
with the domain consisting of functions $\Psi\in H^2((0,\infty)\times\mathbb{R})\oplus H^2((-\infty, 0)\times\mathbb{R})$ satisfying the appropriate matching condition at $x=0$, namely
 % -------------- %
\begin{eqnarray}
  \Psi(0+,y)-\Psi(0-,y) = \frac{\beta}{y} \frac{\partial \Psi}{\partial x}(0+,y)\,,\label{eq-cc1}\\
  \frac{\partial \Psi}{\partial x}(0+,y) = \frac{\partial \Psi}{\partial x}(0-,y)\,.\label{eq-cc2}
\end{eqnarray}
 % -------------- %
The swap $\beta \to -\beta$ is equivalent to the change $y \to -y$ and hence it does not influence the spectrum.
Putting aside the trivial case $\beta=0$, we will therefore assume that $\beta> 0$ in the following.

%\subsection{Main results}

Let $\mathfrak{m}_\mathrm{ac}$ denote the multiplicity function of the absolutely continuous spectra (see e.g. \cite[Chap. 7, Sec. 3 -- 5]{BS}). Our main results can be then stated as follows.

 % -------------- %
\begin{theorem}\label{thm-ac3} (absolutely continuous spectrum of the operators $\mathbf{H}_0$ and $\mathbf{H}_\beta$) \\%[.3em]
The spectrum of operator $\mathbf{H}_0$ is purely absolutely continuous, $\sigma(\mathbf{H}_0) = [\frac12,\infty)$ with $\mathfrak{m}_\mathrm{ac}(E,\mathbf{H}_0) = 2n$ for $E\in(n-\frac12,n+\frac12)$, $\,n\in \mathbb{N}$.
\\%[.3em]
For $\beta>2\sqrt{2}$ the absolutely continuous spectrum of $\mathbf{H}_\beta$ coincides with the spectrum of $\mathbf{H}_0$. For $\beta \leq 2\sqrt{2}$ there is a new branch of continuous spectrum added to the spectrum of $\mathbf{H}_0$ according to Theorems~\ref{thm-ac1} and~\ref{thm-ac2}. For $\beta = 2\sqrt{2}$ we have $\sigma(\mathbf{H}_\beta) = [0,\infty)$ and for $\beta < 2\sqrt{2}$ the spectrum covers the whole real line.
\end{theorem}
 % -------------- %
\begin{theorem}\label{thm-dis} (discrete spectrum of the operator $\mathbf{H}_\beta$ for $\beta\in (2\sqrt{2},\infty)$) \\%[.3em]
Assume $\beta\in (2\sqrt{2},\infty)$, then the discrete spectrum of $\mathbf{H}_\beta$ is nonempty and lies in the interval $(0,\frac12)$. The number of eigenvalues is approximately given by
 % -------------- %
$$
  \frac{1}{4\sqrt{2\left(\frac{\beta}{2\sqrt2} - 1\right)}}\quad \mathrm{as}\quad \beta\to2\sqrt{2}+\,.
$$
 % -------------- %
\end{theorem}
 % -------------- %
\begin{theorem}\label{thm-oneeig} (discrete spectrum of the operator $\mathbf{H}_\beta$ for large $\beta$)
\\%[.3em]
For large enough $\beta$ there is a single eigenvalue which asymptotically behaves as
 % -------------- %
$$
  \Lambda_1 = \frac{1}{2} - \frac{4}{\beta^4} + \mathcal{O}\left(\beta^{-5}\right)\,.
$$
 % -------------- %
\end{theorem}

%%%%%%%%%%%%%%%%%%%%%%%%%%%%%%%%%%%%%
\section{Bound on the quadratic form}
%%%%%%%%%%%%%%%%%%%%%%%%%%%%%%%%%%%%%

It is straightforward to check that the operator $\mathbf{H}_\beta$ with a fixed $\beta>0$ defined above is self-adjoint and the quadratic form $\mathbf{a}_\beta[\Psi] = \mathbf{a}_0[\Psi]+\frac{1}{\beta}\mathbf{b}[\Psi]$
 % -------------- %
\begin{eqnarray*}
  \mathbf{a}_0[\Psi] \!=\! \int_{\mathbb{R}^2} \left(\left|\frac{\partial \Psi}{\partial x}\right|^2+\frac{1}{2}\left|\frac{\partial \Psi}{\partial y}\right|^2+\frac{1}{2}y^2 |\Psi|^2\right)\,\mathrm{d}x\mathrm{d}y\,,\quad \mathbf{b}[\Psi] \!=\! \int_{\mathbb{R}}y \left|\Psi(0+,y)-\Psi(0-,y)\right|^2 \,\mathrm{d}y
\end{eqnarray*}
 % -------------- %
is associated with it. The domain $D = \mathrm{dom\,}\mathbf{a}_0$ of the form $\mathbf{a}_0$ is
 % -------------- %
$$
  D = \left\{\Psi\in H^1((0,\infty)\times\mathbb{R})\oplus H^1((-\infty, 0)\times\mathbb{R})\,; \mathbf{a}_0[\Psi]<\infty\right\}
$$
 % -------------- %
Mimicking the reasoning of \cite{BaE2, BEL} one can check that
the form $\mathbf{a}_0$ is closed on $D$; note that the corresponding
self-adjoint operator separates variables and the $x$-part
describes the motion on line which is free except for the
Neumann condition at $x=0$. Furthermore, using the bounds
\eqref{eq-bound1} and \eqref{eq-bound2} below one can prove in analogy with 
\cite[Proposition 2.2]{BaE2} that $\mathbf{a}_\beta$ is closed on $D$.

Let us add that the quadratic form method cannot be applied
for $\beta < 2\sqrt{2}$ when the spectrum is unbounded from
below. It that case we can proceed as in the case of the usual
Smilansky model \cite{BaE2} establishing the existence of a self-adjoint
Hamiltonian using commutativity with a suitable conjugation.

Our first task is to prove the following bound:
 % -------------- %
\begin{theorem}\label{thm-bound}
If $\beta \geq 2\sqrt{2}$ it holds
 % -------------- %
$$
  \mathbf{a}_\beta [\Psi] \geq \frac{1}{2}\left(1-\frac{2\sqrt{2}}{\beta}\right)\|\Psi\|^2\,.
$$
 % -------------- %
\end{theorem}
 % -------------- %
Before proving this theorem, we provide two technical lemmata:
 % -------------- %
\begin{lemma}\label{lem-1}
For complex numbers $c$, $d$ it holds $2 |\mathrm{Re\,(\bar{c}d})|\leq |c|^2+|d|^2$.
\end{lemma}
 % -------------- %
\begin{proof}
We have
 % -------------- %
$$
  |c\pm d|^2 \geq 0 \quad \Rightarrow \quad  |c|^2+|d|^2\geq \mp(\bar{c}d+c\bar{d})
$$
 % -------------- %
which gives the claim.
\end{proof}
 % -------------- %
\begin{lemma}\label{lem-2}
It holds
 % -------------- %
\begin{eqnarray*}
  \gamma(|\psi(0+)|^2+|\psi(0-)|^2) \leq \int_{\mathbb{R}}\left(|\psi'(x)|^2+\gamma^2|\psi(x)|^2\right)\,\mathrm{d}x\\
 \forall \psi\in H^1((0,\infty))\oplus H^1((-\infty,0))\,,\quad \gamma>0\,,
\end{eqnarray*}
 % -------------- %
with the equality attained on the subspace generated by
 % -------------- %
\begin{equation}
  \tilde \psi_\gamma (x) = \frac{\mathrm{sgn}\,x}{\sqrt{2\gamma}}\mathrm{e}^{-\gamma|x|}\,, \quad\gamma>0\,.\label{eq-utilde}
\end{equation}
 % -------------- %
\end{lemma}
 % -------------- %
\begin{proof}
We have
 % -------------- %
\begin{multline*}
  0\leq \int_{-\infty}^0 |\psi'(x)-\gamma \psi(x)|^2\,\mathrm{d}x = \int_{-\infty}^0 (|\psi'(x)|^2+\gamma^2 |\psi(x)|^2)\,\mathrm{d}x
  -\gamma\int_{-\infty}^0(\bar{\psi}'(x)\psi(x)+\psi'(x)\bar{\psi}(x))\,\mathrm{d}x \\ = \int_{-\infty}^0 (|\psi'(x)|^2+\gamma^2 |\psi(x)|^2)\,\mathrm{d}x -\gamma[|\psi(x)|^2]_{-\infty}^{0-} \,,
\end{multline*}
 % -------------- %
and therefore
 % -------------- %
$$
  \int_{-\infty}^0 (|\psi'(x)|^2+\gamma^2 |\psi(x)|^2)\,\mathrm{d}x  \geq \gamma |\psi(0-)|^2\,.
$$
 % -------------- %
Similarly,
$$
  0\leq \int_{0}^\infty |\psi'(x)+\gamma \psi(x)|^2\,\mathrm{d}x
$$
 % -------------- %
implies
 % -------------- %
$$
  \int_{0}^\infty (|\psi'(x)|^2+\gamma^2 |\psi(x)|^2)\,\mathrm{d}x  \geq \gamma |\psi(0+)|^2\,,
$$
 % -------------- %
and combining both inequalities one obtains the result.
\end{proof}
 % -------------- %

\begin{proof}[Proof of Theorem~\ref{thm-bound}]
We will follow the construction in \cite{Sol1} and obtain a similar bound to the quadratic form $\mathbf{a}_0$. We use separation of variables and the expansion of $\Psi$ in the harmonic oscillator basis, i.e. Hermite functions in the variable $y$ normalized in $L^2(\mathbb{R})$. This yields
 % -------------- %
\begin{equation}
  \Psi(x,y) = \sum_{n\in \mathbb{N}_0} \psi_n(x) \chi_n(y)\,,\label{eq-expansion}
\end{equation}
 % -------------- %
where the symbol $\mathbb{N}_0$ stands for non-negative integers. Inserting the expansion into the form $\mathbf{a}_0$ and using  the known spectrum of the harmonic oscillator we find
 % -------------- %
\begin{equation}
  \mathbf{a}_0[\Psi] = \sum_{n\in\mathbb{N}_0}\int_{\mathbb{R}} \left(|\psi_n'(x)|^2+\left(n+\frac{1}{2}\right)|\psi_n(x)|^2\right)\,\mathrm{d}x \label{eq-a0}
\end{equation}
 % -------------- %
Using twice more the expansion~\eqref{eq-expansion} in combination with the relations satisfied by Hermite functions,
 % -------------- %
\begin{equation}
  \sqrt{n+1}\chi_{n+1}(y)-\sqrt{2}y\chi_n(y) + \sqrt{n}\chi_{n-1}(y) = 0\,,\quad n\in \mathbb{N}_0\,, \label{eq-hermite}
\end{equation}
 % -------------- %
we obtain	
 % -------------- %
\begin{multline}
  \mathbf{b}[\Psi] = \frac{1}{\sqrt{2}}\int_{\mathbb{R}} \sum_{m\in \mathbb{N}_0}\sum_{n\in \mathbb{N}_0} (\bar{\psi}_m(0+)-\bar{\psi}_m(0-))\bar{\chi}_m(y)(\psi_n(0+)-\psi_n(0-)) 
	 \left[\sqrt{n+1}\chi_{n+1}(y)+\sqrt{n}\chi_{n-1}(y)\right]\,\mathrm{d}y \\ = \frac{1}{\sqrt{2}} \sum_{n\in\mathbb{N}_0} [(\bar{\psi}_{n+1}(0+)-\bar{\psi}_{n+1}(0-))\sqrt{n+1}
 + (\bar{\psi}_{n-1}(0+)-\bar{\psi}_{n-1}(0-))\sqrt{n}](\psi_n(0+)-\psi_n(0-)) =
\\ = \frac{2}{\sqrt{2}}\sum_{n\in\mathbb{N}}\sqrt{n}\, \mathrm{Re\,}[(\bar{\psi}_n(0+)-\bar{\psi}_n(0-))(\psi_{n-1}(0+)-\psi_{n-1}(0-))]\,.\label{eq-b}
\end{multline}
 % -------------- %
We employed the Hermite functions orthonormality here and in the last line we have changed the summation index, $n+1\to n$, in the first part of the sum. It follows from Lemma~\ref{lem-1} that
 % -------------- %
$$
  |\mathbf{b}[\Psi]| \leq \frac{1}{\sqrt{2}} \sum_{n\in\mathbb{N}} \sqrt{n} (|\psi_n(0+)-\psi_n(0-)|^2+|\psi_{n-1}(0+)-\psi_{n-1}(0-)|^2)\,.
$$
 % -------------- %
Changing in turn the summation index in the second part of the sum we get
 % -------------- %
$$
  |\mathbf{b}[\Psi]|\leq \frac{1}{\sqrt{2}} \sum_{n\in \mathbb{N}_0} (\sqrt{n}+\sqrt{n+1})|\psi_n(0+)-\psi_n(0-)|^2
  \leq\sum_{n\in\mathbb{N}_0}\sqrt{2n+1} |\psi_n(0+)-\psi_n(0-)|^2\,,
$$
 % -------------- %
where we have used the inequality $\sqrt{n}+\sqrt{n+1}<\sqrt{2(2n+1)}$. Using sub\-sequently Lemmata~\ref{lem-1} and \ref{lem-2} we obtain
 % -------------- %
\begin{equation}
  |\mathbf{b}[\Psi]|\leq 2\sqrt{2}\sum_{n\in\mathbb{N}_0}\sqrt{n+\frac{1}{2}}\left(|\psi_n(0+)|^2+|\psi_n(0-)|^2\right)
 \leq 2\sqrt{2}\sum_{n\in\mathbb{N}_0}\int_{\mathbb{R}}\left(|\psi_n'(x)|^2+\left(n+\frac{1}{2}\right)|\psi_n(x)|^2\right)\,\mathrm{d}x = 2\sqrt{2}\,\mathbf{a}_0[\Psi]\,.\label{eq-bound1}
\end{equation}
 % -------------- %
In this way we arrive, using the fact that $\mathbf{a}_0[\Psi] \geq \frac{1}{2}\|\Psi\|^2$, which follows from \eqref{eq-a0}, at the bound
 % -------------- %
\begin{equation}
  \mathbf{a}_\beta[\Psi] = \mathbf{a}_0[\Psi]+\frac{1}{\beta}\mathbf{b}[\Psi]\geq \left(1-\frac{2\sqrt{2}}{\beta}\right)\,\mathbf{a}_0[\Psi]\geq \frac{1}{2}\left(1-\frac{2\sqrt{2}}{\beta}\right)\|\Psi\|^2\,,\label{eq-bound2}
\end{equation}
 % -------------- %
which proves the theorem; it means, in particular, that the quadratic form associated with $\mathbf{H}_\beta$ is positive definite for $\beta>2\sqrt{2}$.
\end{proof}

%%%%%%%%%%%%%%%%%%%%%%%%%%%%%
\section{The Jacobi operator}
%%%%%%%%%%%%%%%%%%%%%%%%%%%%%

Next we will show that our problem can be rephrased in terms of a Jacobi operator closely related to that used in  \cite{NS}, the two differ only in the parameters involved. We start from eq.~\eqref{eq-cc1}, into which we substitute the Ansatz \eqref{eq-expansion} for $\Psi$, multiply the equation by $\bar{\chi}_m(y)$ and integrate with respect to $y$ over $\mathbb{R}$. Using the orthonormality, we find
 % -------------- %
$$
  \sum_{n\in\mathbb{N}_0}\int_{\mathbb{R}} \bar{\chi}_m(y) y (\psi_n(0+)-\psi_n(0-))\chi_n(y)\,\mathrm{d}y = \beta \sum_{n\in\mathbb{N}_0}\int_{\mathbb{R}} \frac{\partial \psi_n}{\partial x}(0+)\bar{\chi}_m(y)\chi_n(y)\,\mathrm{d}y
$$
 % -------------- %
and relation \eqref{eq-hermite} then yields the condition
 % -------------- %
\begin{eqnarray}
  \beta \frac{\partial \psi_m}{\partial x}(0+) & =& \sum_{n\in\mathbb{N}_0} \frac{1}{\sqrt{2}} \int_{\mathbb{R}} (\psi_n(0+)-\psi_n(0-))\bar{\chi}_m(y) \left(\sqrt{n+1}\chi_{n+1}(y)+\sqrt{n}\chi_{n-1}(y)\right)\,\mathrm{d}y \nonumber \\ &=& \frac{\sqrt{m}}{\sqrt{2}} (\psi_{m-1}(0+)-\psi_{m-1}(0-)) + \frac{\sqrt{m+1}}{\sqrt{2}} (\psi_{m+1}(0+)-\psi_{m+1}(0-))\,,\label{cc-3}
\end{eqnarray}
 % -------------- %
which characterizes the solution `jump' at the axis $x=0$. On the other hand, the condition~\eqref{eq-cc2} implies
 % -------------- %
\begin{equation}
  \frac{\partial \psi_n}{\partial x}(0+) = \frac{\partial \psi_n}{\partial x}(0-) \label{cc-4}
\end{equation}
 % -------------- %
for the coefficient functions. Consider now the eigenvalue problem for the operator $\mathbf{H}_\beta$, which is equivalent to the set of equations
 % -------------- %
\begin{equation}
  - \phi_n''(x)+(n+\frac12-\Lambda) \phi_n(x) = 0\,,\quad x=0\,,\quad n\in \mathbb{N}_0 \label{eq-un}
\end{equation}
 % -------------- %
under the matching conditions~\eqref{cc-3} and \eqref{cc-4}, for $\phi_n\restriction\mathbb{R}_{\pm}\in H^2(\mathbb{R}_{\pm})$ where $\Lambda$ is the sought eigenvalue.

We define $\zeta_n(\Lambda) = \sqrt{n+\frac12-\Lambda}$ taking the branch of the square root which is analytic in $\mathbb{C}\backslash [n+\frac12,\infty)$ and for number~$\Lambda$ from this set it holds
 % -------------- %
$$
  \mathrm{Re\,}\zeta_n(\Lambda) >0\,,\quad \mathrm{Im\,}\zeta_n(\Lambda)\cdot \mathrm{Im\,}\Lambda<0\,.
$$
 % -------------- %
Clearly, solutions to the equation~\eqref{eq-un} in $L^2(\mathbb{R}_{\pm})$ are
 % -------------- %
$$
  \phi_n(x,\Lambda) = k_1(\Lambda)\, \mathrm{e}^{-\zeta_n(\Lambda)x}\,,\quad x>0\,,\quad  \phi_n(x,\Lambda) = k_2(\Lambda)\, \mathrm{e}^{\zeta_n(\Lambda)x}\,,\quad x<0\,,
$$
 % -------------- %
where from \eqref{cc-4} we have $k_1(\Lambda) = -k_2(\Lambda)$. Using the similar normalization as in \cite{NS} we can write $\phi_n(x,\Lambda) = C_n\eta_n(x,\Lambda)$ with
 % -------------- %
$$
  \eta_n(x,\Lambda) := \pm \left(n+\textstyle{\frac12}\right)^{1/4}\mathrm{e}^{\mp\zeta_n(\Lambda)x}\,.\quad x\in\mathbb{R}_{\pm}\,.
$$
 % -------------- %
Hence
 % -------------- %
\begin{eqnarray}
  \phi_n(0+,\Lambda) - \phi_n(0-,\Lambda) = 2 C_n \left(n+\textstyle{\frac12}\right)^{1/4}\,,\nonumber\\
  \frac{\partial \phi_n}{\partial x}(0+,\Lambda) = -C_n\left(n+\textstyle{\frac12}\right)^{1/4} \zeta_n(\Lambda)\,.\label{eq-derphi}
\end{eqnarray}
 % -------------- %
Substituting from here to eq.~\eqref{cc-3} we obtain the relation
 % -------------- %
\begin{equation}
  (n+1)^{1/2}\left(n+\textstyle{\frac32} \right)^{1/4}C_{n+1}+2\mu \left(n+\textstyle{\frac12} \right)^{1/4}\zeta_n(\Lambda)C_n+
   n^{1/2} \left(n-\textstyle{\frac12} \right)^{1/4} C_{n-1} = 0\,,\quad n\in \mathbb{N}_0 \label{eq-jaceq}
\end{equation}
 % -------------- %
with $\mu := \frac{\beta}{2\sqrt{2}}$. This is the same equation as in \cite{NS} and therefore it defines the same Jacobi operator $\mathcal{J}(\Lambda,\mu)$, only our parameter $\mu$ differs from the one used there.

%%%%%%%%%%%%%%%%%%%%%%%%%%%%%%%%%%%%%%%%%
\section{Representation of the resolvent}
%%%%%%%%%%%%%%%%%%%%%%%%%%%%%%%%%%%%%%%%%

Next we will prove a Krein-type formula analogous to eq. (6.6) in \cite{NS}. First we denote $\Psi_\beta\sim \{\psi_{\beta,n}\} = (\mathbf{H}_{\beta}-\Lambda)^{-1}F$, where $F\sim \{f_n\}\in \mathfrak{H} = \ell^2(\mathbb{N}_0,L^2(\mathbb{R}))$; the symbol $\sim$ represents equivalence between elements of $L^2(\mathbb{R}^2)$ and sequences of the coefficient functions. In the `free' case functions $\psi_{0,n}$ satisfy the equation
 % -------------- %
$$
  -\psi''_{0,n} + (n+\textstyle{\frac12}-\Lambda) \psi_{0,n} = f_n
$$
 % -------------- %
and belong to $H^2(\mathbb{R})$. Using integration by parts and the fact that $\psi_{0,n}$ is continuous at zero one finds that
 % -------------- %
$$
  (n+\textstyle{\frac12})^{-1/4}\int_{\mathbb{R}} \eta_n(t,\Lambda)(-\psi''_{0,n}(t))\,\mathrm{d}t
 =2 \psi'_{0,n}(0)-(n+\textstyle{\frac12}-\Lambda)(n+\textstyle{\frac12})^{-1/4}\int_{\mathbb{R}} \eta_n(t,\Lambda)\psi_{0,n}(t)\,\mathrm{d}t\,.
$$
 % -------------- %
Defining
 % -------------- %
$$
  J_n:= \int_{\mathbb{R}}\eta_n(t,\Lambda)f_n(t)\,\mathrm{d}t = (f_n,\eta_n(\cdot,\bar\Lambda))
$$
 % -------------- %
one infers from the previous equation that
 % -------------- %
\begin{equation}
  J_n = 2 (n+\textstyle{\frac12})^{1/4}\psi'_{0,n}(0)\,.\label{eq-jn}
\end{equation}
 % -------------- %

Let now $\Psi_{\beta}-\Psi_0 \sim\{\phi_n\}$, then all the $\phi_n$ have to satisfy the homogenous equation \eqref{eq-un} and their halfline components would belong to $H^2(\mathbb{R}_{\pm})$, which means that $\phi_n(x) = C_n \eta_n(x,\Lambda)$. Now defining
 % -------------- %
$$
  X_n := \textstyle{\frac12}(n+\textstyle{\frac12})^{-1/4}(\psi_{\beta,n}(0+) - \psi_{\beta,n}(0-))
$$
 % -------------- %
we find that
 % -------------- %
\begin{equation}
  X_n -C_n = \textstyle{\frac12}(n+\textstyle{\frac12})^{-1/4}(\psi_{0,n}(0+) - \psi_{0,n}(0-)) = 0\,.\label{eq-xncn}
\end{equation}
 % -------------- %
Substituting for $\psi_{\beta,n}$ into eq.~\eqref{cc-3} we get
 % -------------- %
$$
  \beta \frac{\partial \phi_n}{\partial x}(0+) + \beta \frac{\partial \psi_{0,n}}{\partial x}(0+) = \sqrt{\frac{n}{2}}\, 2 (n-\textstyle{\frac12})^{1/4}X_{n-1} + \sqrt{\frac{n+1}{2}} \,2 (n+\textstyle{\frac32})^{1/4}X_{n+1}\,.
$$
 % -------------- %
Using eqs.~\eqref{eq-derphi}, \eqref{eq-jn}, and \eqref{eq-xncn} we obtain
 % -------------- %
$$
  (n+1)^{1/2}(n+\textstyle{\frac32})^{1/4} X_{n+1}+ 2 \frac{\beta}{2\sqrt{2}}(n+\textstyle{\frac12})^{1/4}\zeta_n(\Lambda) X_n
 +n^{1/2}(n-\textstyle{\frac12})^{1/4}X_{n-1} = \frac{\beta}{2\sqrt{2}}(n+\textstyle{\frac12})^{-1/4} J_n \,.
$$
 % -------------- %
This can be written as
 % -------------- %
$$
  d_{n+1}X_{n+1}+2\mu y_n(\Lambda)X_n +d_n X_{n-1} = \mu J_n
$$
 % -------------- %
with $d_n := n^{1/2}(n+\frac12)^{1/4}(n-\frac12)^{1/4}$, $y_n := (n+\frac12)^{1/2}\zeta_n(\Lambda)$, and $\mu = \frac{\beta}{2\sqrt{2}}$. This is the same non-homogeneous equation as in \cite{NS}, only the parameter $\mu$ and the functions $\eta_n$ are defined differently. Using Jacobi operator $\mathcal{J}$ defined by the left-hand side of the previous equation we can write
 % -------------- %
\begin{equation}
  \mathcal{J}(\Lambda,\mu) X = \mu \{J_n\}\,,\quad X = \{X_n\}\,. \label{eq-j}
\end{equation}
 % -------------- %
We define the operator
 % -------------- %
$$
  \mathbf{T}(\Lambda): \ell^2(\mathbb{N}_0) \to \mathfrak{H}\,,\quad \mathbf{T}(\Lambda)\{X_n\} \sim \{X_n\eta_n(\cdot,\Lambda)\}\,.
$$
 % -------------- %
It is bounded and has a bounded inverse. Its adjoint is
 % -------------- %
$$
  \mathbf{T}(\Lambda)^*: \mathfrak{H} \to \ell^2(\mathbb{N}_0)\,,\quad \mathbf{T}(\Lambda)^*F = \left\{\int_{\mathbb{R}}f_n(x)\eta_n(x,\bar\Lambda)\,\mathrm{d}x\right\}\,,\quad F \sim\{f_n\}\,.
$$
 % -------------- %
Replacing the $\Lambda$ by $\bar\Lambda$ we obtain
 % -------------- %
$$
  \mathbf{T}(\bar\Lambda)^*F = \left\{\int_{\mathbb{R}}f_n(x)\eta_n(x,\Lambda)\,\mathrm{d}x\right\} = \{(f_n,\eta_n(\cdot,\bar\Lambda))\}\,.
$$
 % -------------- %

This makes it possible to express the resolvent of $\mathbf{H}_\beta$ in the following way analogous to Theorem 6.1 in \cite{NS}:
 % -------------- %
\begin{theorem}
Let $\beta>0$, $\mu = \frac{\beta}{2\sqrt{2}}$, and $\Lambda \not\in \mathbb{R}$. Then
 % -------------- %
$$
  (\mathbf{H}_\beta-\Lambda)^{-1} - (\mathbf{H}_0-\Lambda)^{-1} = \mathbf{T}(\Lambda) \mu \mathcal{J}(\Lambda,\mu)^{-1}\mathbf{T}(\bar\Lambda)^{*}\,.
$$
 % -------------- %
\end{theorem}
 % -------------- %
\begin{proof}
From \eqref{eq-j} we find that $X = \mu \mathcal{J}(\Lambda,\mu)^{-1}\mathbf{T}(\bar\Lambda)^* F$. This yields
 % -------------- %
$$
  \Psi_\beta -\Psi_0 = \Phi \sim \{\phi_n\} = \mathbf{T}(\Lambda) \{C_n\} = \mathbf{T}(\Lambda) \{X_n\} = \mathbf{T}(\Lambda) \mu \mathcal{J}(\Lambda,\mu)^{-1}\mathbf{T}(\bar\Lambda)^* F\,,
$$
 % -------------- %
which implies the claim of the theorem.
\end{proof}

%%%%%%%%%%%%%%%%%%%%%%%%%%%%%%%%%%%%%%%%%%%%%%%%%%%%%%%%%%%%%%
\section{Absolutely continuous spectrum of $\mathbf{H}_\beta$}
%%%%%%%%%%%%%%%%%%%%%%%%%%%%%%%%%%%%%%%%%%%%%%%%%%%%%%%%%%%%%%

As usual is the situation when the resolvent allows for a Krein-type formula representation, the spectrum due to the perturbation is encoded in the `denominator', i.e. the operator $\mathcal{J}(\Lambda,\mu)$. In particular, one can use it to find the absolutely continuous spectrum of the operator $\mathbf{H}_\beta$. The argument is no way simple, but since the Jacobi operator involved is, up the modification mentioned, the same as in \cite{NS} one can easily adapt the considerations of that paper (see also Theorem~3.1 there) to arrive at the following conclusions:
 % -------------- %
\begin{theorem}\label{thm-ac1}
\begin{eqnarray*}
	\sigma_{\mathrm{ac}}(\mathbf{H}_\beta) &\!=\!& \sigma_\mathrm{ac}(\mathbf{H}_0)\cup \sigma_\mathrm{ac}(\mathcal{J}_0(\beta/(2\sqrt{2})))\,,\\
	\mathfrak{m}_\mathrm{ac}(E,\mathbf{H}_\beta) &\!=\!& \mathfrak{m}_\mathrm{ac}(E,\mathbf{H}_0) + \mathfrak{m}_\mathrm{ac}(E,\mathcal{J}_0(\beta/(2\sqrt{2}))) \,.
\end{eqnarray*}
where
 % -------------- %$$
$$
  \mathcal{J}_0 (\mu) := D \mathcal{S}+\mathcal{S}^* D + 2\mu Y_0
$$
 % -------------- %
with
\begin{eqnarray*}
  \mathcal{D}\,, \mathcal{S}: \ell^2(\mathbb{N}_0) \mapsto \ell^2(\mathbb{N}_0)\,,\quad 
  \mathcal{D} \{\omega_n\} : \{r_0, r_1,\dots\} \mapsto \{\omega r_0, \omega_1 r_1,\dots\}\,,\\
  D := \mathcal{D}(d_n)\,,\quad Y_0:=\mathcal{D}\{n+1/2\}\,,\quad 
  \mathcal{S} : \{r_0,r_1, \dots\} \mapsto \{0,r_0,r_1,\dots\}\,.
\end{eqnarray*}
We recall that $d_n := n^{1/2}(n+\frac12)^{1/4}(n-\frac12)^{1/4}$.
\end{theorem}
 % -------------- %
\begin{theorem}\label{thm-ac2}
\begin{eqnarray*}
	\sigma(\mathcal{J}_0(\mu)) &\!=\!& (-\infty,\infty)\quad \mathrm{for}\quad \mu<1\,,\\
	\sigma(\mathcal{J}_0(1)) &\!=\!& [0,\infty)\,,\\
	\sigma_\mathrm{ac}(\mathcal{J}_0(\mu)) &\!=\!& \emptyset\quad \mathrm{for}\quad \mu>1\,,\\
	\mathfrak{m}_\mathrm{ac}(E,\mathcal{J}_0(\mu)) &\!=\!& 1 \quad \mathrm{a.e.\ on }\quad \sigma(\mathcal{J}_0(\mu))\,.
\end{eqnarray*}
\end{theorem}
 % -------------- %
Since we have $\mu = \frac{\beta}{2\sqrt{2}}$, these two theorem in combination with the well-known spectrum of $\mathbf{H}_0$ prove the claim of Theorem~\ref{thm-ac3}.

%%%%%%%%%%%%%%%%%%%%%%%%%%%%%%%%%%%%%%%%%%%%%%%%%
\section{Discrete spectrum of $\mathbf{H}_\beta$}
%%%%%%%%%%%%%%%%%%%%%%%%%%%%%%%%%%%%%%%%%%%%%%%%%

The previous results tell us that the discrete spectrum of $\mathbf{H}_\beta$ can exist only in the subcritical situation, $\beta>2\sqrt{2}$. The following discussion is a counterpart of the subcritical case analysis of the original Smilansky model \cite{Sol1}.

We denote conventionally $N_+(\lambda,\mathbf{A}) := \mathrm{dim\,}E^{\mathbf{A}}(\lambda,\infty)\mathcal{H}$, where $E^{\mathbf{A}}(\cdot)$ is the spectral measure of $\mathbf{A}$, and similarly $N_-(\lambda,\mathbf{A})$ is the dimension of the spectral projection to the interval $(-\infty,\lambda)$. If $N_-(\lambda,\mathbf{A})$ is finite, the spectrum of $\mathbf{A}$ in the interval $(-\infty,\lambda)$ is discrete and the number of eigenvalues in this interval with the multiplicity taken into account is equal to $N_-(\lambda,\mathbf{A})$; similarly for $N_+$.

 % -------------- %
\begin{theorem}
Suppose that $\mu = \frac{\beta}{2\sqrt{2}}$ with $\beta\in(2\sqrt{2},\infty)$, then
 % -------------- %
$$
   N_-(\textstyle{\frac12}-\varepsilon,\mathbf{H}_\beta) = N_+(\mu,\mathbf{J}(\varepsilon)) = N_-(-\mu,\mathbf{J}(\varepsilon))
$$
 % -------------- %
holds for $\varepsilon \in (0,\textstyle{\frac12})$, where $\mathbf{J}(\varepsilon)$ is the Jacobi operator in $\ell^2(\mathbb{N}_0)$ generated by a matrix with zero diagonal and non-zero entries
 % -------------- %
$$
  j_{n,n-1}(\varepsilon) = j_{n-1,n}(\varepsilon) = \frac{n^{1/2}}{2(n+\varepsilon)^{1/4}(n-1+\varepsilon)^{1/4}}\,,\quad n\in \mathbb{N}\,.
$$
 % -------------- %
\end{theorem}
 % -------------- %
\begin{proof}
In analogy with \cite[Theorem 3.1]{Sol1} we employ the variational principle by which
 % -------------- %
\begin{eqnarray}
  N_-(\textstyle{\frac12}-\varepsilon,\mathbf{H}_\beta) = \mathop{\mathrm{max}}_{\mathcal{F}\in\mathfrak{F}(\varepsilon)}\,\mathrm{dim\,}\mathcal{F}\,,\label{eq-N-}
\end{eqnarray}
 % -------------- %
where $\mathfrak{F}(\varepsilon)$ is the set of all subspaces $\mathcal{F}\subset D$ such that
 % -------------- %
\begin{equation}
  \mathbf{a}_\beta[\Psi] - (\textstyle{\frac12}-\varepsilon)\|\Psi\|^2_{L^2(\mathbb{R}^2)}<0 \label{eq-ineq}
\end{equation}
 % -------------- %
holds for all nonzero $\Psi\in \mathcal{F}$. We define
 % -------------- %
$$
  \|\Psi\|^2_\varepsilon := \sum_{n\in\mathbb{N}_0} \int_\mathbb{R} (|\psi_n'(x)|^2+(n+\varepsilon)|\psi_n|^2)\,\mathrm{d}x\,,\quad \Psi\sim\{\psi_n\}\,;
$$
 % -------------- %
using relation \eqref{eq-b} one can rewrite the condition \eqref{eq-ineq} as
 % -------------- %
\begin{equation}
  \|\Psi\|^2_\varepsilon + \frac{\sqrt{2}}{\beta}\sum_{n\in\mathbb{N}}\sqrt{n}\, \mathrm{Re\,}[(\bar{\psi}_n(0+)-\bar{\psi}_n(0-))(\psi_{n-1}(0+)-\psi_{n-1}(0-))] < 0\,.\label{eq-ineq2}
\end{equation}
 % -------------- %
Next we define the subspace $\tilde D(\varepsilon) \subset D$ as the set of all
 % -------------- %
$$
  \tilde \Psi \sim \{C_n\tilde \psi_{\sqrt{n+\varepsilon}}\}\,,\quad \{C_n\}\in \ell^2(\mathbb{N}_0)
$$
 % -------------- %
with functions $\tilde \psi_{\gamma}$ introduced in \eqref{eq-utilde}. One can simply check that $\|\Psi\|_\varepsilon :=\sqrt{\|\Psi\|^2_\varepsilon}$ is a norm which satisfies the parallelogram law, hence it induces an inner product,  and moreover, $\|\tilde \Psi\|_\varepsilon = \|\{C_n\}\|_{\ell^2}$. Hence one can define the projection $\Pi_\varepsilon$ onto $\tilde D(\varepsilon)$, orthogonal with respect to the mentioned inner product. For $\Psi\sim\{\psi_n\}\in D$ we then have
 % -------------- %
$$
  \tilde \Psi_\varepsilon := \Pi_\varepsilon \Psi \sim \{C_n\tilde \psi_{\sqrt{n+\varepsilon}}\}
$$
with
 % -------------- %
\begin{equation}
  C_n :=  \int_\mathbb{R} [\psi_n'(x)\tilde \psi'_{\sqrt{n+\varepsilon}}(x)+(n+\varepsilon)\psi_n(x)\tilde \psi_{\sqrt{n+\varepsilon}}(x)]\,\mathrm{d}x  =   (\psi_n(0+)-\psi_n(0-))\frac{(n+\varepsilon)^{1/4}}{\sqrt{2}}\,. \label{eq-cn}
\end{equation}
 % -------------- %
From eqs. \eqref{eq-utilde} and \eqref{eq-cn} we infer that
 % -------------- %
$$
  C_n(\tilde \psi_{\sqrt{n+\varepsilon}}(0+) - \tilde \psi_{\sqrt{n+\varepsilon}}(0-)) = \psi_n(0+) - \psi_n(0-)\,.
$$
 % -------------- %
Now we can argue similarly as in \cite{Sol1}. If we replace $\Psi$ by $\tilde \Psi_\varepsilon$ in the inequality \eqref{eq-ineq2}, the first term does not increase (since $\tilde \Psi_\varepsilon$ is a projection of $\Psi$) and the second term does not change in view of the last displayed equation. Hence the inequality is still valid for $\tilde \Psi_\varepsilon$. Therefore, if $\mathcal{F}\subset D$ belongs to $\mathfrak{F}(\varepsilon)$ then $\Pi_\varepsilon\mathcal{F}$ belongs to $\mathfrak{F}(\varepsilon)$ too. Suppose that there are two subspaces $\mathcal{F}, \mathcal{F}' \in \mathfrak{F}(\varepsilon)$ such that $\mathcal{F}\subset \mathcal{F}'$ and $\mathcal{F}\subset \tilde D(\varepsilon)$. If there exists an element $\Psi\in \mathcal{F}'$ orthogonal to $\tilde D(\varepsilon)$ with respect to the inner product induced by the norm $\|\Psi\|_\varepsilon$, then \eqref{eq-cn} implies $\psi_n(0+)= \psi_n(0-)$, $\forall n \in \mathbb{N}_0$, hence $\mathbf{b}[\Psi] = 0$ and the inequality \eqref{eq-ineq2} is not satisfied. This is a contradiction, so $\mathcal{F}'\subset \tilde D(\varepsilon)$. Thus we can rewrite \eqref{eq-N-} as
 % -------------- %
\begin{equation}\label{}
  N_-(\textstyle{\frac12}-\varepsilon,\mathbf{H}_\beta) = \mathop{\mathrm{max}}_{\mathcal{F}\in\mathfrak{F}(\varepsilon), \mathcal{F}\subset\tilde D(\varepsilon)}\,\mathrm{dim\,}\mathcal{F}\,.\label{eq-varp2}
\end{equation}
 % -------------- %
For each $\tilde \Psi\sim \{C_n\tilde \psi_{\sqrt{n+\varepsilon}}\}\in \tilde D(\varepsilon)$ we obtain using \eqref{eq-cn}
 % -------------- %
$$
  \|\tilde \Psi\|_\varepsilon^2 + \frac{1}{\beta}\mathbf{b}[\tilde \Psi] = \sum_{n\in \mathbb{N}_0} |C_n|^2 + \frac{4\sqrt{2}}{\beta}\sum_{n\in \mathbb{N}}j_{n,n-1}(\varepsilon)\mathrm{Re}\,(C_n\overline{C_{n-1}}) 
 =  \|g\|_{\ell^2(\mathbb{N}_0)}^2 + \frac{1}{\mu}(\mathbf{J}(\varepsilon )g,g)_{\ell^2} = ((I+\mu^{-1}\mathbf{J}(\varepsilon))g,g)
$$
 % -------------- %
with $\mu = \beta/(2\sqrt{2})$ and $g = \{C_n\}\in \ell^2(\mathbb{N}_0)$. The claim of the theorem now follows from \eqref{eq-varp2} and the symmetry of the spectrum of $\mathbf{J}(\varepsilon)$.
\end{proof}

This theorem does not apply if $\varepsilon = 0$ since $j_{1,0} = \infty$. However, it is sufficient to restrict the quadratic form to the subspace $\{g = \{C_n\}: C_0 = 0\}$ of codimension one. As a result of such an operation, the number of eigenvalues is changed by at most one. The limit $\varepsilon \to 0$ the leads to the Jacobi operator $\mathbf{J}_0$ the nonzero entries of which are
 % -------------- %
$$
  j_{n,n-1} = j_{n-1,n} = \frac{1}{2(1-n^{-1})^{1/4}} \,,\quad n\in \mathbb{N}\backslash \{1\}\,,
$$
 % -------------- %
then in analogy with \cite[Theorem 3.2]{Sol1} we arrive at the following conclusion:
 % -------------- %
\begin{theorem}\label{thm-sol32}
  Let $\mu = \beta/(2\sqrt{2})$ with $\beta\in (2\sqrt{2},\infty)$. Then either $N_-(\textstyle{\frac12},\mathbf{H}_\beta) = N_+(\mu,\mathbf{J}_0)$ or $N_-(\textstyle{\frac12},\mathbf{H}_\beta) = N_+(\mu,\mathbf{J}_0) + 1$.
\end{theorem}
 % -------------- %
 
The above two theorems allow us to reduce the task to investigation of the spectral properties of the operator $\mathbf{J}_0$. We are particularly interested what happens with the discrete spectrum when $\beta$ approaches the critical value. The behavior of $N_+(\mu,\mathbf{J}_0)$ as $\mu \to 1+$ is given by \cite[Theorem 3.3]{Sol1}, which we for the reader's convenience we reproduce here.
 % -------------- %
\begin{theorem}\label{thm-sol33}
  Let $\mathbf{J}$ be a zero-diagonal Jacobi matrix the non-diagonal entries of which are
 % -------------- %
$$
  j_{n,n-1} = j_{n-1,n} = \frac{1}{2} +\frac{q}{n}(1+o(1))\,,
$$
 % -------------- %
where $q$ is a positive constant. Then the operator $\mathbf{J}$ has infinitely many non-degenarate eigenvalues $\pm \lambda_k(\mathbf{J})$ with
 % -------------- %
$$
  \lambda_k(\mathbf{J}) = 1+ \frac{2q^2}{k^2}(1+o(1))\quad \mathrm{as}\quad k\to \infty\,,
$$
 % -------------- %
or equivalently,
 % -------------- %
$$
  N_+(\mu,\mathbf{J}) \sim \frac{q\sqrt{2}}{\sqrt{\mu-1}}\quad \mathrm{as}\quad \mu\to 1+\,.
$$
 % -------------- %
\end{theorem}
 % -------------- %

Since our $\mathbf{J}_0$ is a Jacobi matrix of the mentioned type with $q = \frac18$ we obtain from Theorems~\ref{thm-sol32} and~\ref{thm-sol33} the following result which, in combination with Theorems~\ref{thm-bound}, concludes the proof of Theorem~\ref{thm-dis}.
 % -------------- %
\begin{theorem}\label{thm-noe}
Let $\mu = \beta/(2\sqrt{2})$ with $\beta\in (2\sqrt{2},\infty)$, then
 % -------------- %
$$
  N_-(\textstyle{\frac12},\mathbf{H}_\beta)\sim \frac{1}{4\sqrt{2(\mu - 1)}} = \frac{1}{4\sqrt{2\left(\frac{\beta}{2\sqrt2} - 1\right)}}\quad \mathrm{as}\quad \beta\to2\sqrt{2}+\,.
$$
 % -------------- %
\end{theorem}
 % -------------- %

Finally, let us look at the opposite asymptotic regime in the subcritical case and examine the discrete spectrum for weak $\delta'$ coupling. To find the energy gap for large $\beta$ we employ a construction similar to the one in \cite{ELT1,ELT2}.

\begin{proof}[Proof of Theorem~\ref{thm-oneeig}]
First we check that the spectrum on $(-\infty,\textstyle{\frac12})$ is non-empty using a variational argument; the idea is to construct an element $\Psi^\varepsilon \in D$ such that $\mathbf{a}_\beta[\Psi^\varepsilon] <\frac12\|\Psi^\varepsilon\|^2$. Consider functions $\psi_0, \psi_1$ satisfying the conditions
 % -------------- %
$$
   \psi_0(0+)-\psi_0(0-) = -C < 0\,,\quad   \quad \psi_1(0+)-\psi_1(0-) = 1\,,
$$
 % -------------- %
and such that $\Psi=\{\psi_0,\psi_1,0,0,\dots\} \in D$. We scale the first one, $\psi_0^\varepsilon(x) := \psi_0(\varepsilon x)$, and put $\Psi^\varepsilon := \{\psi_0^\varepsilon, \psi_1,0,0, \dots\}$ which belongs again to $D$. From \eqref{eq-a0} and \eqref{eq-b} we have
 % -------------- %
$$
  \mathbf{a}_\beta[\Psi^\varepsilon]-\frac{1}{2}\|\Psi^\varepsilon\|^2 = \int_\mathbb{R}\left(|{\psi_0^\varepsilon}'(x)|^2+|\psi_1(x)|^2+|\psi_1'(x)|^2\right)\,\mathrm{d}x - \frac{\sqrt{2}}{\beta} C 
 = \int_\mathbb{R}\left(\varepsilon|\psi_0'(x)|^2+|\psi_1(x)|^2+|\psi_1'(x)|^2\right)\,\mathrm{d}x - \frac{\sqrt{2}}{\beta} C
$$
 % -------------- %
Obviously, choosing $\varepsilon$ small enough and $C$ large enough one can achieve that the right-hand side of the last equation is negative, which means that the spectrum below $\frac12$ is nonempty for any $\beta>0$.

Using further this conclusion, Theorem~\ref{thm-sol32}, and the fact that the eigenvalues of $\mathbf{J}_0$ have a single accumulation point at $1$ (and consequently, there is a $\mu$ such that there are no eigenvalues of $\mathbf{J}_0$ larger than $\mu$) we find that for $\beta$ large enough the operator $\mathbf{H}_\beta$ has exactly one simple eigenvalue.

The asymptotic expansion of this eigenvalue $\Lambda_1$ can be found by an argument similar to that used in \cite{ELT1} for the original Smilansky model. The system of equations \eqref{eq-jaceq} can be after substitution $Q_n = (n+\textstyle{\frac12})^{1/4}C_n$ rewritten as
 % -------------- %
\begin{eqnarray}
Q_1 + 2\mu \sqrt{\textstyle{\frac{1}{2}}-\Lambda_1} Q_0  &\!=\!& 0\,,\label{jaceq1}\\[.5em]
(n+1)^{1/2}Q_{n+1} + 2\mu \zeta_n(\Lambda_1)Q_n + n^{1/2}Q_{n-1} &\!=\!& 0\,, \quad n\in \mathbb{N}\,. \label{jaceq2}
\end{eqnarray}
 % -------------- %
We normalize $\|Q\|:=\sum_{n =0}^\infty |Q_n|^2 = 1$, using then $\sqrt{n}\leq \sqrt{n+\textstyle{\frac12}-\Lambda_1} = \zeta_n(\Lambda_1)$ and $\sqrt{n+1}\leq \sqrt{2(n+\textstyle{\frac12}-\Lambda_1)}$ we obtain from \eqref{jaceq2} the estimate
 % -------------- %
\begin{eqnarray}
  |Q_n| \leq \frac{1}{2\mu} |Q_{n-1}| + \frac{1}{\sqrt{2}\mu}|Q_{n+1}|\,. \label{eq-qn1}
\end{eqnarray}
 % -------------- %
In the analogy with Lemma~\ref{lem-1} we have
 % -------------- %
$$
  |Q_n|^2 \leq \frac{1}{2\mu^2} |Q_{n-1}|^2 + \frac{1}{\mu^2}|Q_{n+1}|^2\,,
$$
 % -------------- %
and therefore
 % -------------- %
$$
  \sum_{n = 1}^\infty |Q_n|^2 \leq \frac{1}{2\mu^2} \sum_{n=0}^\infty|Q_{n}|^2 + \frac{1}{\mu^2}\sum_{n= 2}^\infty |Q_{n}|^2 \leq \frac{3}{2\mu^2}\,,
$$
 % -------------- %
where we have used the mentioned normalization. From here it follows that
 % -------------- %
\begin{equation}
  |Q_0| = \left(\sum_{n = 0}^\infty |Q_n|^2 - \sum_{n = 1}^\infty |Q_n|^2 \right)^{1/2} = 1+ \mathcal{O}\left(\mu^{-2}\right)\,.\label{eq-q0}
\end{equation}
 % -------------- %
Without loss of generality we may suppose that $Q_0$ is positive. From \eqref{eq-qn1} with $n = 2$ with the use of the normalization we obtain
 % -------------- %
$$
  |Q_2| \leq \frac{1}{2\mu} + \frac{1}{\sqrt{2}\mu}\quad \Rightarrow\quad Q_2 = \mathcal{O}\left(\mu^{-1}\right)\,.
$$
 % -------------- %
Furthermore, from \eqref{jaceq2} and \eqref{eq-q0} we get
 % -------------- %
$$
  Q_1 = \frac{1}{2\mu}+\mathcal{O}\left(\mu^{-2}\right)\,.
$$
 % -------------- %
Finally, from \eqref{jaceq1} we obtain $\left(\frac{1}{2}-\Lambda_1\right)^{1/2} = -\frac{Q_1}{2\mu Q_0} = -\frac{1}{4\mu^2} + \mathcal{O}\left(\mu^{-3}\right)$, or equivalently
 % -------------- %
$$
  \frac{1}{2}-\Lambda_1 = \frac{1}{16\mu^4}+\mathcal{O}\left(\mu^{-5}\right) = \frac{4}{\beta^4} + \mathcal{O}\left(\beta^{-5}\right)\,,
$$
 % -------------- %
which concludes the proof.
\end{proof}

%%%%%%%%%%%%%%%%%%%%%%%%%
\section{Acknowledgement}
%%%%%%%%%%%%%%%%%%%%%%%%%
The research was supported by the Czech Science Foundation (GA\v{C}R) within the project No. 17-01706S. J. L. was also supported by the research programme ``Applied Mathematics'' of the Faculty of Science of the University of Hradec Kr\'alov\'e. The authors are
obliged to the reviewer for suggestions that helped to improve
the presentation.

\end{document}